\documentclass[12pt,a4paper]{amsart}

\usepackage{amsmath,amsfonts,amssymb,mathtools}
\usepackage[hmargin=2cm,vmargin=2cm]{geometry}
\usepackage{hyperref}
\usepackage{nameref,zref-xr}                    
\usepackage{amsthm}
\usepackage{enumerate}

\newcommand{\mbC}{\mathbb C}
\newcommand{\oM}{\overline{\mathcal M}}

\def\d{{\partial}}

\newcommand{\eps}{\varepsilon}
\newcommand{\Coef}{\mathrm{Coef}}

\newcommand{\mcF}{\mathcal{F}}

\newcommand{\hcA}{\widehat{\mathcal{A}}}
\renewcommand{\top}{\mathrm{top}}

\renewcommand{\top}{\mathrm{top}}
\newcommand{\un}{{1\!\! 1}}

\newcommand{\tmcF}{\widetilde{\mathcal{F}}}
\newcommand{\cA}{\mathcal{A}}
\newcommand{\intt}{\mathrm{int}}
\newcommand{\boun}{\mathrm{boun}}
\newcommand{\full}{\mathrm{full}}

\newtheorem{theorem}{Theorem}[section]

\newtheorem{lemma}[theorem]{Lemma}

\newtheorem{expectation}[theorem]{Expectation}

\theoremstyle{remark}

\newtheorem{remark}[theorem]{Remark}

\theoremstyle{definition}

\newtheorem{definition}[theorem]{Definition}

\numberwithin{equation}{section}

\begin{document}

\title{Open topological recursion relations in genus $1$ and integrable systems}

\author{Oscar Brauer}
\address{O. Brauer:\newline
School of Mathematics, University of Leeds, \newline
Leeds, LS2 9JT, United Kingdom}
\email{mmobg@leeds.ac.uk}

\author{Alexandr Buryak}
\address{A. Buryak:\newline 
Faculty of Mathematics, National Research University Higher School of Economics, \newline
6 Usacheva str., Moscow, 119048, Russian Federation; and \newline 
Center for Advanced Studies, Skolkovo Institute of Science and Technology,\newline
1 Nobel str., Moscow, 143026, Russian Federation}
\email{aburyak@hse.ru}

\begin{abstract}
The paper is devoted to the open topological recursion relations in genus $1$, which are partial differential equations that conjecturally control open Gromov--Witten invariants in genus $1$. We find an explicit formula for any solution analogous to the Dijkgraaf--Witten formula for a descendent Gromov--Witten potential in genus $1$. We then prove that at the approximation up to genus $1$ the exponent of an open descendent potential satisfies a system of explicitly constructed linear evolutionary PDEs with one spatial variable. 
\end{abstract}

\date{\today}

\maketitle

\section{Introduction}

{\it Total descendent potentials}, also called {\it formal Gromov--Witten potentials}, are certain formal power series of the form
$$
\mcF(t^*_*,\eps)=\sum_{g\ge 0}\eps^{2g}\mcF_g(t^*_*)\in\mbC[[t^*_*,\eps]],
$$
where $N\ge 1$, and $t^\alpha_a$, $1\le\alpha\le N$, $a\ge 0$, and $\eps$ are formal variables, appearing in various curve counting theories in algebraic geometry including Gromov--Witten theory, Fan--Jarvis--Ruan--Witten theory, and the more recent theory of Gauged Linear Sigma Models. The number~$N$ is often called the {\it rank}. Typically, the coefficients of total descendent potentials are the integrals of certain cohomology classes over moduli spaces of closed Riemann surfaces with additional structures. The function $\mcF_g$ controls the integrals over the moduli spaces of Riemann surfaces of genus $g$. The simplest example of a total descendent potential is the Witten generating series $\mcF^W(t_0,t_1,t_2,\ldots,\eps)$ of intersection numbers on the moduli space of stable Riemann surfaces of genus $g$ with $n$ marked points $\oM_{g,n}$. Note that here and below we omit the upper indices in the $t$-variables when the rank is $1$.

\medskip

There is a unified approach to total descendent potentials using the notion of a cohomological field theory (CohFT)~\cite{KM94} and the Givental group action~\cite{Giv01a,Giv01b,Giv04}. Briefly speaking, the generating series of correlators of CohFTs form the space of {\it total ancestor potentials}, and then using the lower-triangular Givental group action one gets the whole space of total descendent potentials (see e.g. \cite[Section~2]{Sha09} and~\cite{FSZ10}). 

\medskip

There is a remarkable and deep relation between total descendent potentials and the theory of nonlinear PDEs. One of its manifestations is the following system of PDEs for the descendent potential in genus $0$ (see e.g. \cite[Section~2]{Sha09} and~\cite[Corollary 4.13]{FSZ10}):
\begin{align}
&\frac{\d\mcF_0}{\d t^1_0}=\sum_{a\ge 0}t^\alpha_{a+1}\frac{\d\mcF_0}{\d t^\alpha_a}+\frac{1}{2}\eta_{\alpha\beta}t^\alpha_0 t^\beta_0,\label{eq:closed string-0}\\
&\frac{\d^3\mcF_0}{\d t^\alpha_{a+1}\d t^\beta_b\d t^\gamma_c}=\frac{\d^2\mcF_0}{\d t^\alpha_a\d t^\mu_0}\eta^{\mu\nu}\frac{\d^3\mcF_0}{\d t^\nu_0\d t^\beta_b\d t^\gamma_c},\quad 1\le\alpha,\beta,\gamma\le N,\quad a,b,c\ge 0,\label{eq:closed TRR-0}
\end{align}
called the {\it string equation} and the {\it topological recursion relations in genus $0$}, respectively. Here $(\eta_{\alpha\beta})=\eta$ is an $N\times N$ symmetric nondegenerate matrix with complex coefficients, the constants~$\eta^{\alpha\beta}$ are defined by $(\eta^{\alpha\beta}):=\eta^{-1}$, and we use the Einstein summation convention for repeated upper and lower Greek indices. Note that the system of equations~\eqref{eq:closed TRR-0} can be equivalently written~as
\begin{gather*}
d\left(\frac{\d^2\mcF_0}{\d t^\alpha_{a+1}\d t^\beta_b}\right)=\frac{\d^2\mcF_0}{\d t^\alpha_a\d t^\mu_0}\eta^{\mu\nu}d\left(\frac{\d^2\mcF_0}{\d t^\nu_0\d t^\beta_b}\right),\quad 1\le\alpha,\beta\le N,\quad a,b\ge 0,
\end{gather*}
where $d \left( \cdot \right)$ denotes the full differential. 

\medskip

There are equations similar to~\eqref{eq:closed TRR-0} in genus~$1$ (see e.g. \cite[Equation~(1.7)]{EGX00}):
\begin{gather}\label{eq:closed TRR-1}
\frac{\d\mcF_1}{\d t^\alpha_{a+1}}=\frac{\d^2\mcF_0}{\d t^\alpha_a\d t^\mu_0}\eta^{\mu\nu}\frac{\d\mcF_1}{\d t^\nu_0}+\frac{1}{24}\eta^{\mu\nu}\frac{\d^3\mcF_0}{\d t^\mu_0\d t^\nu_0\d t^\alpha_a},\quad 1\le\alpha\le N,\quad a\ge 0.
\end{gather}
They are called the {\it topological recursion relations in genus $1$}. These equations imply that 
\begin{gather}\label{eq:formula for the closed genus 1 potential}
\mcF_1=\frac{1}{24}\log\det(\eta^{-1}M)+G(v^1,\ldots,v^N),
\end{gather}
where the $N\times N$ matrix $M=(M_{\alpha\beta})$ is defined by $M_{\alpha\beta}:=\frac{\d^3\mcF_0}{\d t^1_0\d t^\alpha_0\d t^\beta_0}$, $v^\alpha:=\eta^{\alpha\mu}\frac{\d^2\mcF_0}{\d t^\mu_0\d t^1_0}$, and $G(t^1,\ldots,t^N):=\left.\mcF_1\right|_{t^*_{\ge 1}=0}$ \cite{DW90} (see also \cite[Equation~(1.16)]{DZ98}). Equations similar to~\eqref{eq:closed TRR-0} and~\eqref{eq:closed TRR-1} exist in all genera, but their complexity grow very rapidly with the genus (see e.g.~\cite{Liu07} for some results in genus~$2$).

\medskip

One can see that equations~\eqref{eq:closed string-0},~\eqref{eq:closed TRR-0},~\eqref{eq:closed TRR-1} are universal, meaning that they do not depend on a total descendent potential. On the other hand, there is a rich theory~\cite{DZ01} of hierarchies of evolutionary PDEs with one spatial variable associated to total descendent potentials and containing the full information about these potentials. Conjecturally, for any total descendent potential~$\mcF$ there exists a unique system of PDEs of the form 
\begin{gather}\label{eq:Dubrovin--Zhang system}
\frac{\d w^\alpha}{\d t^\beta_b}=P^\alpha_{\beta,b},\quad 1\le\alpha,\beta\le N,\quad b\ge 0,
\end{gather}
where $w^1,\ldots,w^N\in\mbC[[t^*_*,\eps]]$, $P^\alpha_{\beta,b}$ are differential polynomials in $w^1,\ldots,w^N$, i.e., $P^\alpha_{\beta,b}$ are formal power series in $\eps$ with the coefficients that a polynomials in $w^\gamma_x, w^\gamma_{xx}, \ldots$ (we identify $x=t^1_0$) whose coefficients are formal power series in $w^\gamma$, such that a unique solution of the system~\eqref{eq:Dubrovin--Zhang system} specified by the condition $w^\alpha|_{t^\beta_b=\delta^{\beta,1}\delta_{b,0}x}=\delta^{\alpha,1}x$ is given by $w^\alpha=\eta^{\alpha\mu}\frac{\d^2\mcF}{\d t^\mu_0\d t^1_0}$. This system of PDEs (if it exists) is called the {\it Dubrovin--Zhang hierarchy} or the {\it hierarchy of topological type}. The conjecture is proved at the approximation up to $\eps^2$~\cite{DZ98} and in the case when the Dubrovin--Frobenius manifold associated to the total descendent potential is semisimple~\cite{BPS12a,BPS12b}. The Dubrovin--Zhang hierarchy corresponding to the Witten potential $\mcF^W$ is the Korteweg--de Vries (KdV) hierarchy
\begin{align*}
\frac{\d w}{\d t_1}&=ww_x+\frac{\eps^2}{12}w_{xxx},\\
\frac{\d w}{\d t_2}&=\frac{w^2w_x}{2}+\eps^2\left(\frac{ww_{xxx}}{12}+\frac{w_xw_{xx}}{6}\right)+\eps^4\frac{w_{xxxxx}}{240},\\
&\vdots\notag
\end{align*}
This statement is equivalent to Witten's conjecture~\cite{Wit91}, proved by Kontsevich~\cite{Kon92}.

\medskip

A more recent and less developed field of research is the study of the intersection theory on various moduli spaces of Riemann surfaces with boundary. Such a moduli space always comes with an associated moduli space of closed Riemann surfaces, and, thus, there is the corresponding total descendent potential $\mcF(t^*_*,\eps)=\sum_{g\ge 0}\eps^{2g}\mcF_g(t^*_*)$ of some rank~$N$. There is a large class of examples \cite{PST14,BCT18,ST19,Che18,CZ18,CZ19,Zin20} where the intersection numbers on the corresponding moduli space of Riemann surfaces with boundary of genus~$0$ are described by a formal power series $\mcF^o_0(t^*_*,s_*)\in\mbC[[t^*_*,s_*]]$ depending on an additional sequence of formal variable $s_a$, $a\ge 0$, and satisfying the relations
\begin{align}
&\frac{\d\mcF^o_0}{\d t^1_0}=\sum_{a\ge 0}t^\alpha_{a+1}\frac{\d\mcF^o_0}{\d t^\alpha_a}+\sum_{a\ge 0}s_{a+1}\frac{\d\mcF^o_0}{\d s_a}+s_0,\label{eq:open string in genus 0,1}\\
&d\left(\frac{\d\mcF^o_0}{\d t^\alpha_{a+1}}\right)=\frac{\d^2\mcF_0}{\d t^\alpha_a\d t^\mu_0}\eta^{\mu\nu}d\left(\frac{\d \mcF^o_0}{\d t^\nu_0}\right)+\frac{\d\mcF^o_0}{\d t^\alpha_a}d\left(\frac{\d\mcF^o_0}{\d s_0}\right),&& 1\le\alpha\le N,&& a\ge 0,\label{eq:open TRR-0,t}\\
&d\left(\frac{\d\mcF^o_0}{\d s_{a+1}}\right)=\frac{\d\mcF^o_0}{\d s_a}d\left(\frac{\d\mcF^o_0}{\d s_0}\right),&& && a\ge 0.\label{eq:open TRR-0,s}
\end{align}
Equation~\eqref{eq:open string in genus 0,1} is called the {\it open string equation}. Equations~\eqref{eq:open TRR-0,t} and~\eqref{eq:open TRR-0,s} are called the {\it open topological recursion relations in genus $0$}. The function~$\mcF^o_0$ is called the {\it open descendent potential in genus~$0$}.

\begin{remark}
The system of PDEs~\eqref{eq:open string in genus 0,1}--\eqref{eq:open TRR-0,s} implies that the function $\left.\mcF^o_0\right|_{t^*_{\ge 1}=s_{\ge 1}=0}$ satisfies the open WDVV equations (see \cite[Section~4]{Bur20}), which actually appear in some of the papers mentioned above. However, in~\cite{BB19} the authors presented a construction of an open descendent potential starting from an arbitrary solution of the open WDVV equations.
\end{remark}

Regarding higher genera, much less is known. However, conjecturally, the intersection theory on moduli spaces of Riemann surfaces with boundary of genus $1$ is controlled by formal power series $\mcF^o_1(t^*_*,s_*)\in\mbC[[t^*_*,s_*]]$ satisfying the relations
\begin{align*}
\frac{\d\mcF_1^o}{\d t^\alpha_{a+1}}=&\frac{\d^2\mcF_0}{\d t^\alpha_a\d t^\mu_0}\eta^{\mu\nu}\frac{\d\mcF^o_1}{\d t^\nu_0}+\frac{\d\mcF^o_0}{\d t^\alpha_a}\frac{\d\mcF^o_1}{\d s_0}+\frac{1}{2}\frac{\d^2\mcF_0^o}{\d t^\alpha_a\d s_0},&& 1\le\alpha\le N,&& a\ge 0,\\
\frac{\d\mcF_1^o}{\d s_{a+1}}=&\frac{\d\mcF^o_0}{\d s_a}\frac{\d\mcF^o_1}{\d s_0}+\frac{1}{2}\frac{\d^2\mcF_0^o}{\d s_a\d s_0},&& && a\ge 0,
\end{align*}
called the {\it open topological recursion relations in genus~$1$}. In the case of the intersection theory on the moduli spaces of Riemann surfaces with boundary of genus~$g$ with~$k$ boundary marked points and~$l$ internal marked points~$\oM_{g,k,l}$, these relations were conjectured by the authors of~\cite{PST14} and proved in~\cite[Section~6.2.3]{BCT18} (a proof by other methods is obtained by J.~P.~Solomon and R.~J.~Tessler in a work in preparation). An evidence that the open topological recursion relations in genus $1$ hold for the open $r$-spin theory is also given in~\cite[Section~6.2.3]{BCT18}. 

\medskip

An analog of the theory of Dubrovin--Zhang hierarchies for solutions of the system~\eqref{eq:open string in genus 0,1}--\eqref{eq:open TRR-0,s} was developed in~\cite{BB19}. Regarding higher genera, a very promising direction was opened by the series of papers~\cite{PST14,Tes15,Bur15,Bur16,BT17} (see also~\cite{ABT17}), where the authors studied the intersection numbers on the moduli spaces of Riemann surfaces with boundary of genus~$g$ with~$k$ boundary marked points and~$l$ internal marked points~$\oM_{g,k,l}$.
The main result of these works is the proof~\cite{BT17} of the Pandharipande--Solomon--Tessler conjecture~\cite{PST14} saying that the generating series 
$$
\mcF^{o,PST}(t_*,s_*,\eps)=\sum_{g\ge 0}\eps^g\mcF^{o,PST}_g(t_*,s_*)\in\mbC[[t_*,s_*,\eps]]
$$
of the intersection numbers satisfies the following system of PDEs:
\begin{align}
&\frac{\d}{\d t_p}\exp(\eps^{-1}\mcF^{o,PST})=\frac{\eps^{-1}}{(2p+1)!!}\left(L^{p+\frac{1}{2}}\right)_+\exp(\eps^{-1}\mcF^{o,PST}),&& p\ge 0,\label{eq:Lax for PST,1}\\
&\frac{\d}{\d s_p}\exp(\eps^{-1}\mcF^{o,PST})=\frac{\eps^{-1}}{2^{p+1}(p+1)!}L^{p+1}\exp(\eps^{-1}\mcF^{o,PST}),&& p\ge 0,\label{eq:Lax for PST,2}
\end{align}
where $L=(\eps\d_x)^2+2w$ is the Lax operator for the KdV hierarchy, and $w=\frac{\d^2\mcF^W}{\d t_0^2}$. 
\begin{remark}
To be precise, we have presented a version of the Pandharipande--Solomon--Tessler conjecture, which is slightly different from the original one in two aspects. First of all, in~\cite{PST14} the authors considered a function $\tmcF^{o,PST}$ related to our function $\mcF^{o,PST}$ by $\tmcF^{o,PST}=\left.\mcF^{o,PST}\right|_{s_{\ge 1}=0}$. The function $\mcF^{o,PST}$ can be reconstructed from the function $\tmcF^{o,PST}$ using the system of PDEs 
$$
\frac{\d}{\d s_p}\exp(\eps^{-1}\mcF^{o,PST})=\frac{\eps^p}{(p+1)!}\frac{\d^{p+1}}{\d s_0^{p+1}}\exp(\eps^{-1}\mcF^{o,PST}),\quad p\ge 1.
$$ 
Second, the system of PDEs from~\cite[Conjecture~2]{PST14} determining the function $\tmcF^{o,PST}$ does not have the form of a system of evolutionary PDEs with one spatial variable. The fact that the presented version of the Pandharipande--Solomon--Tessler conjecture is equivalent to the original one was observed in~\cite{Bur16}. 
\end{remark}

In this paper we study solutions of the open topological recursion relations in genus $1$. First, we find an analog of formula~\eqref{eq:formula for the closed genus 1 potential}. Then, using this formula, we construct a system of linear PDEs of the form similar to~\eqref{eq:Lax for PST,1} and~\eqref{eq:Lax for PST,2} such that the function $\exp(\mcF^o_0+\eps\mcF^o_1)$ satisfies it at the approximation up to $\eps$. An expectation in higher genera and a relation with a Lax description of the Dubrovin--Zhang hierarchies are also discussed.

\subsection*{Acknowledgements}

O.~B. is supported by Becas CONACYT para estudios de Doctorado en el extranjero awarded by the Mexican government, Ref: 2020-000000-01EXTF-00096. The work of A.~B. is funded within the framework of the HSE University Basic Research Program and the Russian Academic Excellence Project '5-100'.

\medskip

We are grateful to Oleg Chalykh for valuable remarks about the preliminary version of the paper.


\section{Closed and open descendent potentials in genus $0$}

In this section we recall the definitions of closed and open descendent potentials in genus~$0$ and the construction of associated to them systems of PDEs.

\subsection{Differential polynomials}

Consider formal variables $v^\alpha_i$, $\alpha=1,\ldots,N$, $i=0,1,\ldots$. Following~\cite{DZ01} (see also~\cite{Ros17}) we define the ring of {\it differential polynomials} $\cA_{v^1,\ldots,v^N}$ in the variables $v^1,\ldots,v^N$ as the ring of polynomials in the variables~$v^\alpha_i$, $i>0$, with coefficients in the ring of formal power series in the variables $v^\alpha:= v^\alpha_0$:
$$
\cA_{v^1,\ldots,v^N}:=\mbC[[v^*]][v^*_{\ge 1}].
$$

\begin{remark} It is useful to think of the variables $v^\alpha=v^\alpha_0$ as the components $v^\alpha(x)$ of a formal loop $v\colon S^1\to\mbC^N$ in the standard basis of $\mbC^N$. Then the variables $v^\alpha_1:= v^\alpha_x, v^\alpha_2:= v^\alpha_{xx},\ldots$ are the components of the iterated $x$-derivatives of the formal loop.
\end{remark}

The {\it standard gradation} on $\cA_{v^1,\ldots,v^N}$, which we denote by $\deg$, is introduced by $\deg v^\alpha_i:= i$. The homogeneous component of $\cA_{v^1,\ldots,v^N}$ of standard degree $d$ is denoted by $\cA_{v^1,\ldots,v^N;d}$. Introduce an operator $\d_x\colon\cA_{v^1,\ldots,v^N}\to\cA_{v^1,\ldots,v^N}$ by 
$$
\partial_x := \sum_{i\geq 0} v^\alpha_{i+1}\frac{\d}{\d v^\alpha_i}.
$$
It increases the standard degree by $1$. 

\medskip

Consider the extension $\hcA_{v^1,\ldots,v^N}:= \cA_{v^1,\ldots,v^N}[[\eps]]$ of the space $\cA_{v^1,\ldots,v^N}$ with a new variable~$\eps$ of standard degree $\deg\eps:= -1$. Let $\hcA_{v^1,\ldots,v^N;d}$ denote the subspace of degree~$d$ of $\hcA$. Abusing the terminology we still call elements of the space $\hcA_{v^1,\ldots,v^N}$ \emph{differential polynomials}.

\subsection{Closed descendent potentials in genus $0$}

Let us fix $N\ge 1$, an $N\times N$ symmetric nondegenerate complex matrix $\eta=(\eta_{\alpha\beta})$, and an $N$-tuple of complex numbers $(A^1,\ldots,A^N)$, not all equal to zero. We will use the notation
$$
\frac{\d}{\d t^\un_a}:=A^\alpha\frac{\d}{\d t^\alpha_a},\quad a\ge 0.
$$

\begin{definition}
A formal power series $\mcF_0\in\mbC[[t^*_*]]$ is called a {\it descendent potential in genus $0$} if it satisfies the following system of PDEs:
\begin{align}
&\sum_{a\ge 0} t^\alpha_{a+1}\frac{\d\mcF_0}{\d t^\alpha_a}-\frac{\d\mcF_0}{\d t^\un_0}=-\frac{1}{2}\eta_{\alpha\beta}t^\alpha_0 t^\beta_0,&&&&\label{eq:string for desc-0}\\
&\sum_{a\ge 0} t^\alpha_a\frac{\d\mcF_0}{\d t^\alpha_a}-\frac{\d\mcF_0}{\d t^\un_1}=2\mcF_0,&&&&\label{eq:dilaton for desc-0}\\
&\frac{\d^3\mcF_0}{\d t^\alpha_{a+1}\d t^\beta_b\d t^\gamma_c}=\frac{\d^2\mcF_0}{\d t^\alpha_a\d t^\mu_0}\eta^{\mu\nu}\frac{\d^3\mcF_0}{\d t^\nu_0\d t^\beta_b\d t^\gamma_c},&& 1\le\alpha,\beta,\gamma\le N,&& a,b,c\ge 0,\label{eq:TRR for desc-0}\\
&\frac{\d^2\mcF_0}{\d t^\alpha_{a+1}\d t^\beta_b}+\frac{\d^2\mcF_0}{\d t^\alpha_a\d t^\beta_{b+1}}=\frac{\d^2\mcF_0}{\d t^\alpha_a\d t^\mu_0}\eta^{\mu\nu}\frac{\d^2\mcF_0}{\d t^\nu_0\d t^\beta_b},&& 1\le\alpha,\beta\le N,&& a,b\ge 0.\notag
\end{align}
We will sometimes call a descendent potential in genus~$0$ a {\it closed} descendent potential in genus~$0$ in order to distinguish it from an open analog that we will discuss below.
\end{definition}
\begin{remark}
Doing a linear change of variables, one can make $\frac{\d}{\d t^\un_a}=\frac{\d}{\d t^1_a}\Leftrightarrow A^\alpha=\delta^{\alpha,1}$ in Equations~\eqref{eq:string for desc-0} and~\eqref{eq:dilaton for desc-0}. That is why authors often assume that $A^\alpha=\delta^{\alpha,1}$.
\end{remark}
\begin{remark}
For any total descendent potential $\mcF=\sum_{g\ge 0}\eps^{2g}\mcF_g$ the function $\mcF_0$ is a descendent potential in genus~$0$. However, describing precisely which descendent potentials in genus $0$ can be extended to total descendent potentials is an interesting open problem.
\end{remark}

Define differential polynomials $\Omega^{[0]}_{\alpha,a;\beta,b}\in\cA_{v^1,\ldots,v^N;0}$, $1\le\alpha,\beta\le N$, $a,b\ge 0$, by
$$
\Omega^{[0]}_{\alpha,a;\beta,b}:=\left.\frac{\d^2\mcF_0}{\d t^\alpha_a\d t^\beta_b}\right|_{t^\gamma_c=\delta_{c,0}v^\gamma},
$$
and let
$$
(v^\top)^\alpha:=\eta^{\alpha\mu}\frac{\d^2\mcF_0}{\d t^\mu_0\d t^\un_0}\in\mbC[[t^*_*]],\quad 1\le\alpha\le N.
$$
Then we have (see e.g.~\cite[Proposition~3]{BPS12b})
\begin{gather}\label{eq:property of the closed two-point functions}
\frac{\d^2\mcF_0}{\d t^\alpha_a\d t^\beta_b}=\left.\Omega^{[0]}_{\alpha,a;\beta,b}\right|_{v^\gamma=(v^\top)^\gamma}.
\end{gather}
This implies that the $N$-tuple of functions $\left.(v^\top)^\alpha\right|_{t^\gamma_0\mapsto t^\gamma_0+A^\gamma x}$ is a solution of the following system of PDEs:
$$
\frac{\d v^\alpha}{\d t^\beta_b}=\eta^{\alpha\mu}\d_x\Omega^{[0]}_{\mu,0;\beta,b},\quad 1\le\alpha,\beta\le N,\quad b\ge 0,
$$
which is called the {\it principal hierarchy} associated to the potential~$\mcF_0$.

\subsection{Open descendent potentials in genus $0$}

Let us fix a closed descendent potential in genus $0$ $\mcF_0$.

\begin{definition}\label{definition:open descendent potential in genus 0}
An {\it open descendent potential in genus $0$} $\mcF^o_0\in\mbC[[t^*_*,s_*]]$ is a solution of the following system of PDEs:
\begin{align}
&\sum_{b\ge 0}t^\beta_{b+1}\frac{\d\mcF^o_0}{\d t^\beta_b}+\sum_{a\ge 0}s_{a+1}\frac{\d\mcF^o_0}{\d s_a}-\frac{\d\mcF^o_0}{\d t^\un_0}=-s_0,\label{eq:open string in genus 0,2}\\
&\sum_{b\ge 0}t^\beta_b\frac{\d\mcF^o_0}{\d t^\beta_b}+\sum_{a\ge 0}s_a\frac{\d\mcF^o_0}{\d s_a}-\frac{\d\mcF^o_0}{\d t^\un_1}=\mcF^o_0,\notag\\
&d\left(\frac{\d\mcF^o_0}{\d t^\alpha_{p+1}}\right)=\frac{\d^2\mcF_0}{\d t^\alpha_p\d t^\mu_0}\eta^{\mu\nu}d\left(\frac{\d \mcF^o_0}{\d t^\nu_0}\right)+\frac{\d\mcF^o_0}{\d t^\alpha_p}d\left(\frac{\d\mcF^o_0}{\d s_0}\right),\label{eq:open TRR for desc-0,t}\\
&d\left(\frac{\d\mcF^o_0}{\d s_{p+1}}\right)=\frac{\d\mcF^o_0}{\d s_p}d\left(\frac{\d\mcF^o_0}{\d s_0}\right).\label{eq:open TRR for desc-0,s}
\end{align}
\end{definition}

\medskip

Consider a new formal variable $\phi$. Similarly to the differential polynomials $\Omega^{[0]}_{\alpha,a;\beta,b}$, let us introduce differential polynomials $\Gamma^{[0]}_{\alpha,a},\Delta^{[0]}_a\in\cA_{v^1,\ldots,v^N,\phi;0}$, $1\le\alpha\le N$, $a\ge 0$, by
$$
\Gamma^{[0]}_{\alpha,a}:=\left.\frac{\d\mcF^o_0}{\d t^\alpha_a}\right|_{\substack{t^\gamma_c=\delta_{c,0}v^\gamma\\s_c=\delta_{c,0}\phi}},\qquad \Delta^{[0]}_a:=\left.\frac{\d\mcF^o_0}{\d s_a}\right|_{\substack{t^\gamma_c=\delta_{c,0}v^\gamma\\s_c=\delta_{c,0}\phi}},
$$ 
and let
$$
\phi^\top:=\frac{\d\mcF^o_0}{\d t^\un_0}\in\mbC[[t^*_*,s_*]].
$$
We have the following properties, analogous to the property~\eqref{eq:property of the closed two-point functions} (\cite[Section~4.4]{BB19}, \cite[Proposition~2.2]{ABLR20}):
\begin{gather*}
\frac{\d\mcF^o_0}{\d t^\alpha_a}=\left.\Gamma^{[0]}_{\alpha,a}\right|_{\substack{v^\gamma=(v^\top)^\gamma\\\phi=\phi^\top}},\qquad\frac{\d\mcF^o_0}{\d s_a}=\left.\Delta^{[0]}_{a}\right|_{\substack{v^\gamma=(v^\top)^\gamma\\\phi=\phi^\top}}.
\end{gather*}
This implies that the $(N+1)$-tuple of functions $\left.\left((v^\top)^1,\ldots,(v^\top)^N,\phi^\top\right)\right|_{t^\gamma_0\mapsto t^\gamma_0+A^\gamma x}$ satisfies the following system of PDEs:
\begin{align*}
&\frac{\d v^\alpha}{\d t^\beta_b}=\d_x\eta^{\alpha\mu}\Omega^{[0]}_{\mu,0;\beta,b},&& \frac{\d v^\alpha}{\d s_b}=0,\\
&\frac{\d\phi}{\d t^\beta_b}=\d_x\Gamma^{[0]}_{\beta,b},&& \frac{\d\phi}{\d s_b}=\d_x\Delta^{[0]}_b,
\end{align*}
which we call the {\it extended principal hierarchy} associated to the pair of potentials $(\mcF_0,\mcF^o_0)$.


\section{Open descendent potentials in genus $1$}

Here we introduce the notion of an open descendent potential in genus $1$ and prove two main results of our paper: Theorems~\ref{theorem:explicit formula in the open genus 1} and~\ref{theorem:linear PDEs up to genus 1}.

\subsection{Open descendent potentials in genus $1$}

Let us fix a pair $(\mcF_0,\mcF^o_0)$ of closed and open potentials in genus $0$.

\begin{definition}\label{definition:open descendent potential in genus 1}
An {\it open descendent potential in genus $1$} $\mcF^o_1\in\mbC[[t^*_*,s_*]]$ is a solution of the following system of PDEs:
\begin{align}
\frac{\d\mcF_1^o}{\d t^\alpha_{a+1}}=&\frac{\d^2\mcF_0}{\d t^\alpha_a\d t^\mu_0}\eta^{\mu\nu}\frac{\d\mcF^o_1}{\d t^\nu_0}+\frac{\d\mcF^o_0}{\d t^\alpha_a}\frac{\d\mcF^o_1}{\d s_0}+\frac{1}{2}\frac{\d^2\mcF_0^o}{\d t^\alpha_a\d s_0},&& 1\le\alpha\le N,&& a\ge 0,\label{eq:open TRR-1,1}\\
\frac{\d\mcF_1^o}{\d s_{a+1}}=&\frac{\d\mcF^o_0}{\d s_a}\frac{\d\mcF^o_1}{\d s_0}+\frac{1}{2}\frac{\d^2\mcF_0^o}{\d s_a\d s_0},&& && a\ge 0.\label{eq:open TRR-1,2}
\end{align}
\end{definition}

\medskip

Consider an open descendent potential in genus $1$ $\mcF^o_1$. Define a formal power series $G^o\in\mbC[[v^*,\phi]]$ by
$$
G^o:=\left.\mcF^o_1\right|_{\substack{t^\alpha_a=\delta_{a,0}v^\alpha\\s_a=\delta_{a,0}\phi}}.
$$

\begin{theorem}\label{theorem:explicit formula in the open genus 1}
We have
\begin{gather}\label{eq:formula for the open genus 1 part}
\mcF^o_1=\frac{1}{2}\log\frac{\d^2\mcF^o_0}{\d t^\un_0\d s_0}+\left.G^o\right|_{\substack{v^\gamma=(v^\top)^\gamma\\\phi=\phi^\top}}.
\end{gather}
\end{theorem}
\begin{proof}
Note that Equation~\eqref{eq:open string in genus 0,2} implies that
$$
\left.\frac{\d^2\mcF^o_0}{\d t^\un_0\d s_0}\right|_{t^*_{\ge 1}=s_{\ge 1}=0}=1.
$$
Therefore, the logarithm $\log\frac{\d^2\mcF^o_0}{\d t^\un_0\d s_0}$ is a well-defined formal power series in the variables~$t^*_*$ and~$s_*$. Also, Equations~\eqref{eq:string for desc-0} and~\eqref{eq:open string in genus 0,2} imply that
$$
\left.(v^\top)^\alpha\right|_{t^*_{\ge 1}=0}=t^\alpha_0,\qquad \left.\phi^\top\right|_{t^*_{\ge 1}=s_{\ge 1}=0}=s_0.
$$
Therefore, Equation~\eqref{eq:formula for the open genus 1 part} is true when $t^*_{\ge 1}=s_{\ge 1}=0$. 

\medskip

Using the linear differential operators
\begin{align*}
&P^1_{\alpha,a}:=\frac{\d}{\d t^\alpha_{a+1}}-\frac{\d^2\mcF_0}{\d t^\alpha_a\d t^\mu_0}\eta^{\mu\nu}\frac{\d}{\d t^\nu_0}-\frac{\d\mcF^o_0}{\d t^\alpha_a}\frac{\d}{\d s_0},&&1\le\alpha\le N,&& a\ge 0,\\
&P^2_a:=\frac{\d}{\d s_{a+1}}-\frac{\d\mcF^o_0}{\d s_a}\frac{\d}{\d s_0},&& && a\ge 0,
\end{align*}
Equations~\eqref{eq:open TRR-1,1} and~\eqref{eq:open TRR-1,2} can be equivalently written as
\begin{align*}
&P^1_{\alpha,a}\mcF^o_1=\frac{1}{2}\frac{\d^2\mcF^o_0}{\d t^\alpha_a\d s_0}, && 1\le\alpha\le N, && a\ge 0,\\
&P^2_a\mcF^o_1=\frac{1}{2}\frac{\d^2\mcF^o_0}{\d s_a\d s},&& && a\ge 0.
\end{align*}
This system of PDEs uniquely determines the function $\mcF^o_1$ starting from the initial condition~$\left.\mcF^o_1\right|_{t^*_{\ge 1}=s_{\ge 1}=0}=G^o(t^1_0,\ldots,t^N_0,s_0)$. By Equations~\eqref{eq:TRR for desc-0},~\eqref{eq:open TRR for desc-0,t}, and~\eqref{eq:open TRR for desc-0,s}, we have
$$
P^1_{\alpha,a}(v^\top)^\beta=P^1_{\alpha,a}\phi^\top=P^2_a(v^\top)^\beta=P^2_a\phi^\top=0.
$$
Therefore, it remains to check that
\begin{align}
&P^1_{\alpha,a}\log\frac{\d^2\mcF^o_0}{\d t^\un_0\d s_0}=\frac{\d^2\mcF^o_0}{\d t^\alpha_a\d s_0},\label{eq:proof of open genus 1 formula,1}\\
&P^2_a\log\frac{\d^2\mcF^o_0}{\d t^\un_0\d s_0}=\frac{\d^2\mcF^o_0}{\d s_a\d s_0}.\label{eq:proof of open genus 1 formula,2}
\end{align}

\medskip

To prove Equation~\eqref{eq:proof of open genus 1 formula,1}, we compute
\begin{align*}
P^1_{\alpha,a}\log\frac{\d^2\mcF^o_0}{\d t^\un_0\d s_0}=&\frac{1}{\frac{\d^2\mcF^o_0}{\d t^\un_0\d s_0}}\left(\frac{\d^3\mcF^o_0}{\d t^\alpha_{a+1}\d t^\un_0\d s_0}-\frac{\d^2\mcF_0}{\d t^\alpha_a\d t^\mu_0}\eta^{\mu\nu}\frac{\d^3\mcF^o_0}{\d t^\nu_0\d t^\un_0\d s_0}-\frac{\d\mcF^o_0}{\d t^\alpha_a}\frac{\d^3\mcF^o_0}{\d s_0\d t^\un_0\d s_0}\right)=\\
=&\frac{1}{\frac{\d^2\mcF^o_0}{\d t^\un_0\d s_0}}\left[\frac{\d}{\d s_0}\underline{\left(\frac{\d^2\mcF^o_0}{\d t^\alpha_{a+1}\d t^\un_0}-\frac{\d^2\mcF_0}{\d t^\alpha_a\d t^\mu_0}\eta^{\mu\nu}\frac{\d^2\mcF^o_0}{\d t^\nu_0\d t^\un_0}-\frac{\d\mcF^o_0}{\d t^\alpha_a}\frac{\d^2\mcF^o_0}{\d s_0\d t^\un_0}\right)}+\frac{\d^2\mcF_0^o}{\d t^\alpha_a\d s_0}\frac{\d^2\mcF_0^o}{\d t^\un_0\d s_0}\right]=\\
=&\frac{\d^2\mcF_0^o}{\d t^\alpha_a\d s_0},
\end{align*}
where the vanishing of the underlined expression follows from Equation~\eqref{eq:open TRR for desc-0,t}. The proof of Equation~\eqref{eq:proof of open genus 1 formula,2} is analogous. The theorem is proved.
\end{proof}

\subsection{Differential operators and PDEs}

Consider a differential operator $L$ of the form
$$
L=\sum_{i\ge 0}L_i(v^*_*,\eps)(\eps\d_x)^i,\quad L_i\in\hcA_{v^1,\ldots,v^N;0}.
$$
Let $f$ be a formal variable and consider the PDE
\begin{gather}\label{eq:L-f PDE,1}
\frac{\d}{\d t}\exp(\eps^{-1} f)=\eps^{-1}L\exp(\eps^{-1} f).
\end{gather}
Note that 
$$
\frac{(\eps\d_x)^i\exp(\eps^{-1}f)}{\exp(\eps^{-1}f)}=Q_i(f_*,\eps),\quad i\ge 0,
$$
where $Q_i\in\hcA_f$ can be recursively computed by the relation
$$
Q_i=
\begin{cases}
1,&\text{if $i=0$},\\
f_xQ_{i-1}+\eps\d_x Q_{i-1},&\text{if $i\ge 1$}.
\end{cases}
$$ 

\begin{remark}
Note that $Q_i$ does not depend on $f$ and is a polynomial in the derivatives $f_x,f_{xx},\ldots$ and $\eps$. Moreover, if we introduce a new formal variable~$\psi$ and substitute $f_{i+1}=\psi_i$, $i\ge 0$, then~$Q_i$ becomes a differential polynomial of degree~$0$. 
\end{remark}

We see that PDE~\eqref{eq:L-f PDE,1} is equivalent to the following PDE:
\begin{gather}\label{eq:L-f PDE,2}
\frac{\d f}{\d t}=\sum_{i\ge 0}L_i(v^*_*,\eps)Q_i(f_*,\eps).
\end{gather}

\medskip

Let us look at this PDE in more details at the approximation up to~$\eps$.
\begin{lemma}
We have $Q_i=f_x^i+\eps\frac{i(i-1)}{2}f_x^{i-2}f_{xx}+O(\eps^2)$.
\end{lemma}
\begin{proof}
The formula is clearly true for $i=0$. We proceed by induction:
\begin{align*}
Q_{i+1}=&f_x Q_i+\eps\d_x Q_i=f_x^{i+1}+\eps\left(f_x\frac{i(i-1)}{2}f_x^{i-2}f_{xx}+\d_x\left(f_x^i\right)\right)+O(\eps^2)=\\
=&f_x^{i+1}+\eps\frac{(i+1)i}{2}f_x^{i-1}f_{xx}+O(\eps^2).
\end{align*}
\end{proof}

Consider the expansion
$$
L_i(v^*_*,\eps)=\sum_{j\ge 0}L_i^{[j]}(v^*_*)\eps^j,\quad L_i^{[j]}\in\cA_{v^1,\ldots,v^N;j}.
$$
We see that Equation~\eqref{eq:L-f PDE,2} has the form
\begin{gather}\label{eq:linear PDE - first order approximation}
\frac{\d f}{\d t}=\sum_{i\ge 0}L_i^{[0]}f_x^i+\eps\sum_{i\ge 0}\left(L_i^{[1]}f_x^i+L_i^{[0]}\frac{i(i-1)}{2}f_x^{i-2}f_{xx}\right)+O(\eps^2).
\end{gather}

\subsection{A linear PDE for an open descendent potential up to genus $1$}

Define differential operators $L_{\alpha,a}^{\intt}$, $1\le\alpha\le N$, $a\ge 0$, and $L_a^{\boun}$, $a\ge 0$, by
\begin{align*}
&L_{\alpha,a}^{\intt}:=\sum_{i\ge 0}\left(L_{\alpha,a,i}^{\intt;[0]}+\eps L_{\alpha,a,i}^{\intt;[1]}\right)(\eps\d_x)^i,\\
&L_a^{\boun}:=\sum_{i\ge 0}\left(L_{a,i}^{\boun;[0]}+\eps L_{a,i}^{\boun;[1]}\right)(\eps\d_x)^i,
\end{align*}
where
\begin{align*}
&L_{\alpha,a,i}^{\intt;[0]}:=\Coef_{\phi^i}\Gamma^{[0]}_{\alpha,a}\in\cA_{v^1,\ldots,v^N;0},\\
&L_{\alpha,a,i}^{\intt;[1]}:=\Coef_{\phi^i}\left[\left(\frac{\d G^o}{\d\phi}\frac{\d\Gamma^{[0]}_{\alpha,a}}{\d v^\beta}-\frac{\d G^o}{\d v^\beta}\frac{\d\Gamma^{[0]}_{\alpha,a}}{\d\phi}+\frac{1}{2}\frac{\d^2\Gamma^{[0]}_{\alpha,a}}{\d v^\beta\d\phi}\right)v^\beta_x+\frac{\d G^o}{\d v^\beta}\eta^{\beta\gamma}\d_x\Omega^{[0]}_{\gamma,0;\alpha,a}\right]\in\cA_{v^1,\ldots,v^N;1},\\
&L_{a,i}^{\boun;[0]}:=\Coef_{\phi^i}\Delta^{[0]}_a\in\cA_{v^1,\ldots,v^N;0},\\
&L_{a,i}^{\boun;[1]}:=\Coef_{\phi^i}\left[\left(\frac{\d G^o}{\d\phi}\frac{\d\Delta^{[0]}_a}{\d v^\beta}-\frac{\d G^o}{\d v^\beta}\frac{\d\Delta^{[0]}_a}{\d\phi}+\frac{1}{2}\frac{\d^2\Delta^{[0]}_a}{\d v^\beta\d\phi}\right)v^\beta_x\right]\in\cA_{v^1,\ldots,v^N;1}.
\end{align*}

\begin{theorem}\label{theorem:linear PDEs up to genus 1}
The formal power series $v^\beta=\left.(v^\top)^\beta\right|_{t^\gamma_0\mapsto t^\gamma_0+A^\gamma x}$ and $f=\left.\left(\mcF^o_0+\eps\mcF^o_1\right)\right|_{t^\gamma_0\mapsto t^\gamma_0+A^\gamma x}$ satisfy the system of PDEs
\begin{align}
&\frac{\d}{\d t^\alpha_a}\exp(\eps^{-1} f)=\eps^{-1}L_{\alpha,a}^{\intt}\exp(\eps^{-1} f),&& 1\le\alpha\le N,&& a\ge 0,\label{eq:internal equation}\\
&\frac{\d}{\d s_a}\exp(\eps^{-1} f)=\eps^{-1}L_a^{\boun}\exp(\eps^{-1} f),&& && a\ge 0.\label{eq:boundary equation}
\end{align}
at the approximation up to $\eps$.
\end{theorem}
\begin{proof}
Abusing notations let us denote the formal powers series $\left.\mcF^o_0\right|_{t^\gamma_0\mapsto t^\gamma_0+A^\gamma x}$, $\left.\mcF^o_1\right|_{t^\gamma_0\mapsto t^\gamma_0+A^\gamma x}$, $\left.(v^\top)^\alpha\right|_{t^\gamma_0\mapsto t^\gamma_0+A^\gamma x}$, and $\left.\phi^\top\right|_{t^\gamma_0\mapsto t^\gamma_0+A^\gamma x}$ by $\mcF^o_0$, $\mcF^o_1$, $v^\alpha$, and $\phi$, respectively. We can then write the statement of Theorem~\ref{theorem:explicit formula in the open genus 1} as
$$
\mcF^o_1=\frac{1}{2}\log\phi_s+G^o.
$$

\medskip

Let us prove Equation~\eqref{eq:internal equation} at the approximation up to~$\eps$. We have
$$
\d_x\left(\mcF^o_0+\eps\mcF^o_1\right)=\phi+\eps\left(\frac{1}{2}\frac{\phi_{xs}}{\phi_s}+\d_x G^o\right).
$$
Therefore, by Equation~\eqref{eq:linear PDE - first order approximation}, we have to check that
\begin{align*}
&\frac{\d}{\d t^\alpha_a}\left(\mcF^o_0+\eps\mcF^o_1\right)=\sum_{i\ge 0}L_{\alpha,a,i}^{\intt;[0]}\phi^i+\eps\sum_{i\ge 0}L_{\alpha,a,i}^{\intt;[1]}\phi^i+\\
&\hspace{3.2cm}+\eps\sum_{i\ge 0}L_{\alpha,a,i}^{\intt;[0]}\left(\frac{i(i-1)}{2}\phi^{i-2}\phi_x+i\phi^{i-1}\left(\frac{1}{2}\frac{\phi_{xs}}{\phi_s}+\d_x G^o\right)\right)\Leftrightarrow\\
\Leftrightarrow&\frac{\d}{\d t^\alpha_a}\left(\mcF^o_0+\eps\mcF^o_1\right)=\Gamma^{[0]}_{\alpha,a}+\eps\left(\sum_{i\ge 0}L_{\alpha,a,i}^{\intt;[1]}\phi^i+\frac{1}{2}\frac{\d^2\Gamma^{[0]}_{\alpha,a}}{\d\phi^2}\phi_x+\frac{1}{2}\frac{\d\Gamma^{[0]}_{\alpha,a}}{\d\phi}\frac{\phi_{xs}}{\phi_s}+\frac{\d\Gamma^{[0]}_{\alpha,a}}{\d\phi}\d_x G^o\right)\Leftrightarrow\\
\Leftrightarrow&\frac{\d\mcF^o_1}{\d t^\alpha_a}=\frac{1}{2}\frac{\d\Gamma^{[0]}_{\alpha,a}}{\d\phi}\frac{\phi_{xs}}{\phi_s}+\left(\frac{1}{2}\frac{\d^2\Gamma^{[0]}_{\alpha,a}}{\d\phi^2}+\frac{\d\Gamma^{[0]}_{\alpha,a}}{\d\phi}\frac{\d G^o}{\d\phi}\right)\phi_x+\frac{\d\Gamma^{[0]}_{\alpha,a}}{\d\phi}\frac{\d G^o}{\d v^\beta}v^\beta_x+\sum_{i\ge 0}L_{\alpha,a,i}^{\intt;[1]}\phi^i.
\end{align*}

\medskip

Using the definition of $L^{\intt;[1]}_{\alpha,a,i}$, we see that the last equation is equivalent to
\begin{align}
&\frac{\d\mcF^o_1}{\d t^\alpha_a}=\frac{1}{2}\frac{\d\Gamma^{[0]}_{\alpha,a}}{\d\phi}\frac{\phi_{xs}}{\phi_s}+\left(\frac{1}{2}\frac{\d^2\Gamma^{[0]}_{\alpha,a}}{\d\phi^2}+\frac{\d\Gamma^{[0]}_{\alpha,a}}{\d\phi}\frac{\d G^o}{\d\phi}\right)\phi_x+\notag\\
&\hspace{1.3cm}+\left(\frac{\d G^o}{\d\phi}\frac{\d\Gamma^{[0]}_{\alpha,a}}{\d v^\beta}+\frac{1}{2}\frac{\d^2\Gamma^{[0]}_{\alpha,a}}{\d v^\beta\d\phi}\right)v^\beta_x+\frac{\d G^o}{\d v^\beta}\eta^{\beta\gamma}\d_x\Omega^{[0]}_{\gamma,0;\alpha,a}\Leftrightarrow\notag\\
\Leftrightarrow&\frac{\d\mcF^o_1}{\d t^\alpha_a}=\frac{1}{2}\frac{\d\Gamma^{[0]}_{\alpha,a}}{\d\phi}\frac{\phi_{xs}}{\phi_s}+\frac{1}{2}\d_x\frac{\d\Gamma^{[0]}_{\alpha,a}}{\d\phi}+\frac{\d G^o}{\d\phi}\d_x\Gamma^{[0]}_{\alpha,a}+\frac{\d G^o}{\d v^\beta}\eta^{\beta\gamma}\d_x\Omega^{[0]}_{\gamma,0;\alpha,a}.\label{eq:intermediate equation for the proof of linear PDE}
\end{align}
On the other hand, we compute
\begin{align*}
\frac{\d\mcF^o_1}{\d t^\alpha_a}=&\left(\frac{1}{2}\log\phi_s+G^o\right)_{t^\alpha_a}=\frac{1}{2}\frac{\left(\phi_{t^\alpha_a}\right)_s}{\phi_s}+\frac{\d G^o}{\d v^\beta}\eta^{\beta\gamma}\d_x\Omega^{[0]}_{\gamma,0;\alpha,a}+\frac{\d G^o}{\d\phi}\d_x\Gamma^{[0]}_{\alpha,a}=\\
=&\frac{1}{2}\frac{\d_x\left(\Gamma^{[0]}_{\alpha,a}\right)_s}{\phi_s}+\frac{\d G^o}{\d v^\beta}\eta^{\beta\gamma}\d_x\Omega^{[0]}_{\gamma,0;\alpha,a}+\frac{\d G^o}{\d\phi}\d_x\Gamma^{[0]}_{\alpha,a}=\\
=&\frac{1}{2}\frac{\d_x\left(\frac{\d\Gamma^{[0]}_{\alpha,a}}{\d\phi}\phi_s\right)}{\phi_s}+\frac{\d G^o}{\d v^\beta}\eta^{\beta\gamma}\d_x\Omega^{[0]}_{\gamma,0;\alpha,a}+\frac{\d G^o}{\d\phi}\d_x\Gamma^{[0]}_{\alpha,a}=\\
=&\frac{1}{2}\d_x\frac{\d\Gamma^{[0]}_{\alpha,a}}{\d\phi}+\frac{1}{2}\frac{\d\Gamma^{[0]}_{\alpha,a}}{\d\phi}\frac{\phi_{xs}}{\phi_s}+\frac{\d G^o}{\d v^\beta}\eta^{\beta\gamma}\d_x\Omega^{[0]}_{\gamma,0;\alpha,a}+\frac{\d G^o}{\d\phi}\d_x\Gamma^{[0]}_{\alpha,a},
\end{align*}
which proves Equation~\eqref{eq:intermediate equation for the proof of linear PDE} and, hence, Equation~\eqref{eq:internal equation} at the approximation up to~$\eps$.

\medskip

The proof of Equation~\eqref{eq:boundary equation} is analogous.
\end{proof}


\section{Expectation in higher genera}

Consider a total descendent potential $\mcF(t^*_*,\eps)=\sum_{g\ge 0}\eps^{2g}\mcF_g(t^*_*)$ of some rank $N$.

\begin{expectation}\label{expectation}
Under possibly some additional assumptions, there exists a reasonable geometric construction of an open descendent potential in all genera $\mcF^o(t^*_*,s^*,\eps)=\sum_{g\ge 0}\eps^g\mcF^o_g(t^*_*,\eps)$ satisfying the following properties:
\begin{itemize}
\item The functions $\mcF^o_0$ and $\mcF^o_1$ are open descendents potentials in genus $0$ and $1$, respectively (according to Definitions~\ref{definition:open descendent potential in genus 0} and~\ref{definition:open descendent potential in genus 1}).
\item The function $\mcF^o$ satisfies the open string equation in all genera
$$
\sum_{b\ge 0}t^\beta_{b+1}\frac{\d\mcF^o}{\d t^\beta_b}+\sum_{a\ge 0}s_{a+1}\frac{\d\mcF^o}{\d s_a}-\frac{\d\mcF^o}{\d t^\un_0}=-s_0+C\eps,
$$
where $C$ is some constant.
\item Consider formal variables $w^1,\ldots,w^N$. Then there exist differential operators $L_{\alpha,a}^{\full,\intt}$, $1\le\alpha\le N$, $a\ge 0$, and $L_a^{\full,\boun}$, $a\ge 0$, of the form
\begin{align*}
&L_{\alpha,a}^{\full,\intt}=\sum_{i\ge 0}L_{\alpha,a,i}^{\full,\intt}(w^*_*,\eps)(\eps\d_x)^i, && L_{\alpha,a,i}^{\full,\intt}\in\hcA_{w^1,\ldots,w^n;0},\\
&L_a^{\full,\boun}=\sum_{i\ge 0}L_{a,i}^{\full,\boun}(w^*_*,\eps)(\eps\d_x)^i, && L_{a,i}^{\full,\boun}\in\hcA_{w^1,\ldots,w^n;0},\\
&L_{\alpha,a,i}^{\full,\intt}=\left.L_{\alpha,a,i}^{\intt}\right|_{v^\beta_b=w^\beta_b}+O(\eps^2),\\
&L_{a,i}^{\full,\boun}=\left.L_{a,i}^{\boun}\right|_{v^\beta_b=w^\beta_b}+O(\eps^2),
\end{align*}
such that the formal power series $w^\beta=\left.\eta^{\beta\mu}\frac{\d^2\mcF}{\d t^\mu_0\d t^\un_0}\right|_{t^\gamma_0\mapsto t^\gamma_0+A^\gamma x}$ and $f=\left.\mcF^o\right|_{t^\gamma_0\mapsto t^\gamma_0+A^\gamma x}$ satisfy the system of PDEs
\begin{align}
&\frac{\d}{\d t^\alpha_a}\exp(\eps^{-1} f)=\eps^{-1}L_{\alpha,a}^{\full,\intt}\exp(\eps^{-1} f),&& 1\le\alpha\le N,&& a\ge 0,\label{eq:full internal equation}\\
&\frac{\d}{\d s_a}\exp(\eps^{-1} f)=\eps^{-1}L_a^{\full,\boun}\exp(\eps^{-1} f),&& && a\ge 0.\label{eq:full boundary equation}
\end{align}
\end{itemize}
\end{expectation}

\medskip

Suppose that there exists a Dubrovin--Zhang hierarchy corresponding to our total descendent potential $\mcF$ (this is true when, for example, the associated Dubrovin--Frobenius manifold is semisimple). It is easy to show that if Expectation~\ref{expectation} is true, then the flows $\frac{\d}{\d t^\alpha_a}$ and $\frac{\d}{\d s_b}$ pairwise commute, which means that
\begin{align}
&\frac{\d L^{\full,\intt}_{\alpha,a}}{\d t^\beta_b}-\frac{\d L^{\full,\intt}_{\beta,b}}{\d t^\alpha_a}+\eps^{-1}\left[L^{\full,\intt}_{\alpha,a}, L^{\full,\intt}_{\beta,b}\right]=0, && 1\le\alpha,\beta\le N, &&a,b\ge 0,\notag\\
&\frac{\d L^{\full,\boun}_a}{\d t^\beta_b}+\eps^{-1}\left[L^{\full,\boun}_a, L^{\full,\intt}_{\beta,b}\right]=0, && 1\le\beta\le N, &&a,b\ge 0,\label{eq:full s-t compatibility}\\
&\left[L^{\full,\boun}_a, L^{\full,\boun}_b\right]=0, && && a,b\ge 0,\notag
\end{align}
where the derivatives $\frac{\d L^{\full,\intt}_{\alpha,a}}{\d t^\beta_b}$ and $\frac{\d L^{\full,\boun}_a}{\d t^\beta_b}$ are computed using the flows of the Dubrovin--Zhang hierarchy. Note that Equation~\eqref{eq:full s-t compatibility} potentially gives a Lax description of the Dubrovin--Zhang hierarchy (see an alternative approach in~\cite{CvdLPS14}).

\end{document}